\documentclass[11pt]{article}
\usepackage{comment,amsfonts,amssymb,amsmath,amsthm,graphicx,algorithmic,algorithm,fullpage}
\usepackage{times}

\newtheorem{theorem}{Theorem}
\newtheorem{lemma}{Lemma}
\newtheorem{corollary}[lemma]{Corollary}
\newtheorem{remark}{Remark}

\newtheorem{definition}{Definition}

\newcommand{\namedref}[2]{\hyperref[#2]{#1~\ref*{#2}}}
\newcommand{\sectionref}[1]{\namedref{Section}{#1}}

\newcommand{\subsectionref}[1]{\namedref{Subsection}{#1}}
\newcommand{\theoremref}[1]{\namedref{Theorem}{#1}}
\newcommand{\defref}[1]{\namedref{Definition}{#1}}

\newcommand{\lemmaref}[1]{\namedref{Lemma}{#1}}

\newcommand{\E}{{\mathbb{E}}}
\newcommand{\N}{\mathbb{N}}

\newcommand{\dist}{{\rm dist}}

\newcommand{\polylog}{{\rm polylog}}
\newcommand{\poly}{{\rm poly}}

\ifx\pdftexversion\undefined
\usepackage[colorlinks,linkcolor=black,filecolor=black,citecolor=black,urlco
lor=black,pdfstartview=FitH]{hyperref}
\else
\usepackage[colorlinks,linkcolor=blue,filecolor=blue,citecolor=blue,urlcolor
=blue,pdfstartview=FitH]{hyperref}
\fi

\usepackage{xargs}
\usepackage[prependcaption,textsize=scriptsize]{todonotes}
\newcommandx{\Rtodo}[2][1=]{\todo[linecolor=red,backgroundcolor=red!25,bordercolor=red,#1]{#2}}

\newcommand{\alert}[1]{\textbf{\color{red}
		[[[#1]]]}\marginpar{\textbf{\color{red}**}}\typeout{ALERT:
		\the\inputlineno: #1}}

\floatstyle{ruled}
\newfloat{algorithm}{tbp}{loa}
\providecommand{\algorithmname}{Algorithm}
\floatname{algorithm}{\protect\algorithmname}

\usepackage{pdfsync}
\makeatother

\begin{document}
	\title{On Notions of Distortion and an Almost Minimum Spanning Tree with Constant Average Distortion\footnote{Preliminary version of this paper was published in SODA'16 \cite{BFN16}.}}
	
	\author{Yair Bartal \thanks{School of Engineering and Computer Science, Hebrew University. Email: \texttt{yair@cs.huji.ac.il}. Supported in part by a grant from the Israeli Science Foundation (1817/17).} \\
\and Arnold Filtser \thanks{Department of Computer Science, Ben-Gurion University of the Negev. Email: \texttt{arnoldf@cs.bgu.ac.il}. Supported in part by ISF grant No. (1817/17) and by BSF grant No. 2015813.}\\
\and Ofer Neiman \thanks{Department of Computer Science, Ben-Gurion University of the Negev. Email: \texttt{neimano@cs.bgu.ac.il}. Supported in part by ISF grant No. (1817/17) and by BSF grant No. 2015813.}
	}

	\date{}
	\maketitle

	\begin{abstract}
		This paper makes two main contributions: The first is the construction of a {\sl near-minimum spanning tree} with {\sl constant average distortion}. The second is a general equivalence theorem relating two refined notions of distortion: {\sl scaling distortion} and {\sl prioritized distortion}.
		
		Minimum Spanning Trees of weighted graphs are fundamental objects in numerous applications. In particular in distributed networks, the minimum spanning tree of the network is often used to route messages between network nodes. Unfortunately, while being most efficient in the total cost of connecting all nodes, minimum spanning trees fail miserably in the desired property of approximately preserving distances between pairs.
		While known lower bounds exclude the possibility of the worst case distortion of a tree being small, it was shown in \cite{ABN15} that there exists a spanning tree with constant average distortion. Yet, the weight of such a tree may be significantly larger than that of the MST.
		In this paper, we show that any weighted undirected graph admits a {\em spanning tree} whose weight is at most $(1+\rho)$ times that of the MST, providing {\em constant average distortion} $O(1/\rho)$. Our result exhibits the best possible tradeoff of this type.

		This result makes use of a {\em general equivalence theorem} relating two recently developed notions of distortion for metric embedding.	
		The first is the notion of {\em scaling distortion}, which provides improved distortion for $1-\epsilon$ fractions of the pairs, for all $\epsilon$ simultaneously. A stronger version called {\em coarse} scaling distortion, has improved distortion guarantees for the furthest pairs.
		The second notion is that of {\em prioritized distortion}, a property allowing to prioritize the nodes whose associated distortions will be improved. We show that prioritized distortion is essentially equivalent to coarse scaling distortion via a general transformation.
		
		Our spanning tree result is achieved by first showing the existence of a low weight {\em spanner} with small prioritized distortion, which by our theorem implies scaling distortion, which in turn implies the constant average distortion bound.
		
		The scaling-prioritized distortion equivalence theorem has further implications and is of 
		independent interest. In particular, it allows us to resolve a few questions concerning prioritized embedding. Mainly, we obtain a strengthening of Bourgain's theorem on embedding arbitrary metrics into Euclidean space, possessing optimal prioritized distortion.
\end{abstract}

	\section{Introduction}
	
	One of the fundamental problems in graph theory is that of constructing a Minimum Spanning Tree (MST) of a given weighted graph $G=(V,E)$. This problem and its variants received much attention, and has found numerous applications. In many of these applications, one may desire not only minimizing the weight of the spanning tree, but also other desirable properties, at the price of losing a small factor in the weight of the tree compared to that of the MST. Define the {\em lightness} of $T$ to be the total weight of $T$ (the sum of its edge weights) divided by the weight of an MST.
	One well known example is that of a Shallow Light Tree (SLT) \cite{KRY93,ABP92}, which is a rooted spanning tree having near optimal $(1+\rho)$ lightness, while approximately preserving all distances from the root to the other vertices.
	
	It is natural to ask that the spanning tree will preserve well all pairwise distances in the graph. However, it is easy to see that no spanning tree can maintain such a requirement. In particular, even in the case of the unweighted cycle graph on $n$ vertices, for every spanning tree there is a pair of neighboring vertices whose distance increases by a factor of $n-1$. A natural relaxation of this demand is that the spanning tree approximates all pairwise distances {\em on average}. Formally, the distortion of the pair $u,v\in V$ in $T$ is defined as $\frac{d_T(u,v)}{d_G(u,v)}$, and the {\em average distortion} is $\frac{1}{{n\choose 2}}\sum_{\{u,v\}\in{V\choose 2}}\frac{d_T(u,v)}{d_G(u,v)}$, where $d_G$ (respectively $d_T$) is the shortest-path metric in $G$ (resp. $T$).\footnote{Distortion is sometimes referred to as stretch.}
	In \cite{ABN15}, it was shown that for every weighted graph, it is possible to find a spanning tree which has constant average distortion.
	
	In this paper, we devise a spanning tree of  optimal $(1+\rho)$ lightness that has $O(1/\rho)$ average distortion over all pairwise distances.
	We show that this result is tight by exhibiting a lower bound on the tradeoff between lightness and average distortion, that in order to get $1+\rho$ lightness the average distortion must be $\Omega(1/\rho)$ (this holds even if the spanning subgraph is not necessarily a tree), and in particular, the average distortion for an MST is as bad as $\Omega(n)$.

	Our main result of a light spanning tree with constant average distortion may be of interest for network applications. It is extremely common in the area of distributed computing that an MST is used for communication between the network nodes. This allows easy centralization of computing processes and an efficient way of broadcasting through the network, allowing communication to all nodes at a minimum cost. Yet, as already mentioned above, when communication is required between specific pairs of nodes, the cost of routing through the MST may be extremely high, even when their real distance is small.
	However, in practice it is the average distortion, rather than the worst-case distortion, that is often used as a practical measure of quality, as has been a major motivation behind the initial work of \cite{KSW04,ABN11,ABN15}. As noted above, the MST still fails even in this relaxed measure.
	Our result overcomes this by promising small routing cost between nodes on average, while still possessing the low cost of broadcasting through the tree, thereby maintaining the standard advantages of the MST.
	
	Our main result on a low average distortion embedding follows from analyzing the {\em scaling distortion} \cite{KSW04,ABN11} of the embedding. This refined notion of distortion turns out to be closely related to another useful measure of {\em prioritized distortion} \cite{EFN15}. The second main contribution of this paper is providing an equivalence theorem stating the relation between these useful notions.
	
\subsection{Scaling Distortion vs. Prioritized Distortion: A General Equivalence Theorem}
	{\sl Scaling distortion}, first introduced in \cite{KSW04}\footnote{Originally coined gracefully degrading embedding.}, requires that for {\em every} $0<\epsilon<1$, the distortion of all but an $\epsilon$-fraction of the pairs is bounded by the appropriate function of $\epsilon$. In \cite{ABN11} it was shown that one may obtain bounds on the average distortion, as well as on higher moments of the distortion function, from bounds on the scaling distortion. In \cite{ABN11} several scaling distortion results were shown including $O(\log(1/\epsilon))$ scaling distortion embedding into Euclidean space, and in \cite{ABN15} an $O(1/\sqrt{\epsilon})$ scaling embedding into trees, and spanning trees in particular.
	
	{\sl Prioritized distortion}, introduced recently in \cite{EFN15}, requires that for every given ranking $v_1,\dots,v_n$ of the vertices of the graph, there is an embedding where the distortion of pairs including $v_j$ is bounded as a function of the rank $j$. Several prioritized distortion results were given in \cite{EFN15}, including $\tilde{O}(\log j)$ \footnote{By $\tilde{O}(f(n))$ we mean $O(f(n)\cdot{\rm polylog}(f(n)))$.} prioritized distortion embedding into Euclidean space.
	
	One of the main ingredients of our work is a {\em general reduction}
	relating the notions of prioritized distortion and scaling distortion. In fact, we show that prioritized distortion is essentially equivalent to a strong version of scaling distortion called {\em coarse} scaling distortion, in which for every point $p$ and every $0<\epsilon<1$, the distances to the  $1-\epsilon$ fraction of the farthest points from $p$ are preserved with the desired distortion. We prove that
any embedding with a prioritized distortion $\alpha$
has coarse scaling distortion bounded by $O(\alpha(8/\epsilon))$.
This result could be of independent interest; in particular, it shows that the results of \cite{EFN15}
with coarse scaling distortion $\gamma$, there exists an embedding with prioritized distortion $\gamma(\mu(j))$, where $\mu$ is a function such that $\sum_i \mu(i) = 1$ (e.g. $\mu(j)= \frac{6}{(\pi\cdot j)^2}$).
	We note that this reduction heavily relied on the property of coarse scaling distortion embeddings and does not apply to non-coarse scaling embeddings. Yet, most existing scaling embeddings are indeed coarse.
	This result implies  that all existing {\em coarse} scaling distortion results  
	have priority distortion counterparts, thus improving few of the results of \cite{EFN15}.
	In particular, by applying a theorem of \cite{ABN11} we obtain prioritized embedding of arbitrary metric spaces into $l_p$ in dimension $O(\log n)$ and prioritized distortion $O(\log j)$, which exhibits a strengthening of Bourgain's theorem, and is best possible. It also implies better bounds for decomposable metrics (see \cite{ABN11}), such as planar and doubling metrics, where we obtain an optimal $O(\sqrt{\log j})$ prioritized distortion.

We also show an equivalence between embeddings with coarse partial distortion and terminal embeddings, which can be used to extend and improve previous results. See \sectionref{sec:terminal-partial} for details.
		
	In the context of our main construction of a light spanning tree, the first direction of the above equivalence theorem allows us to devise prioritized distortion embeddings and use these to obtain scaling distortion embeddings which possess the desired constant average distortion.

\subsection{Light Spanning tree of Constant Average Distortion.}

	\sloppy Our main spanning tree construction provides a light spanning tree with scaling distortion bound of $\tilde{O}(1/\sqrt{\epsilon})/\rho$, which is  nearly tight as a function of $\epsilon$ \cite{ABN15}. This result implies that the average distortion is $O(1/\rho)$.
	
	We also devise a probabilistic embedding: a distribution over (light) spanning trees with $\polylog(1/\epsilon)/\rho$ scaling distortion, thus providing constant bounds on all fixed moments of the distortion (i.e., the $l_q$-distortion \cite{ABN11} for fixed $q$).
	
	Our main technical contribution, en route to this result, may be of its own interest: We devise a {\em spanner} (a subgraph of $G$) with $1+\rho$ lightness and low {\em prioritized distortion}.
	Here we show a light spanner construction with prioritized distortion at most $\tilde{O}(\log j)/\rho$.
	Using the equivalence theorem relating prioritized distortion and scaling distortion (discussed above), we obtain a spanner having scaling distortion $\tilde{O}(\log (1/\epsilon))/\rho$, and thus average distortion $O(1/\rho)$. Although we do not obtain a spanning tree here, this result has a few advantages, as we get constant bounds on all fixed moments of the distortion function (the $\ell_q$-distortion). Moreover, the worst-case distortion is only logarithmic in $n$. We note that all of our results admit deterministic polynomial time algorithms.

	Another technical contribution is a general, black-box reduction, that transform constructions of spanners with distortion $t$ and lightness $\ell$ into spanners with distortion $t/\delta$ and lightness $1+\delta\ell$ (here $0<\delta<1$). This reduction can be applied in numerous settings, and also for many different special families of graphs. In particular, this reduction allows us to construct prioritized spanners with lightness arbitrarily close to $1$.

	\paragraph{Outline and Techniques.}
	Our proof has the following high level approach; Given a graph and a ranking of its vertices, we first find a low weight spanner with prioritized distortion $\tilde{O}(\log j)/\rho$. We then apply the general reduction from prioritized distortion to scaling distortion to find a spanner with scaling distortion $\tilde{O}(\log (1/\epsilon))/\rho$. Finally, we use the result of \cite{ABN15} to find a spanning tree of this spanner with scaling distortion $O(1/\sqrt{\epsilon})$. We then conclude that the scaling distortion of the concatenated embeddings is roughly their product, which implies our main result of a spanning tree with lightness $1+\rho$ and scaling distortion $\tilde{O}(1/\sqrt{\epsilon})/\rho$.
	
	Similarly, we can apply the probabilistic embedding of \cite{ABN15} to get a light counterpart, devising a distribution over spanning trees, each with lightness $1+\rho$, with (expected) scaling distortion $\polylog(1/\epsilon)/\rho$.
	
	The main technical part of the paper is finding a light prioritized spanner. In a recent result \cite{CW16} (following \cite{ENS14,CDNS92}), it was shown that any graph on $n$ vertices admits a spanner with (worst-case) distortion $O(\log n)$ and with constant lightness.
	However, these constructions have no bound on the more refined notions of distortion. To obtain a prioritized distortion, we use a technique similar in spirit to \cite{EFN15}: group the vertices into $\log\log n$ sets according to their priority, the set $K_i$ will contain vertices with priority up to $2^{2^i}$. We then build a low weight spanner for each of these sets. As prioritized distortion guarantees a bound for {\em every pair} containing a high ranking vertex, we must augment the spanner of $K_i$ with shortest paths to all other vertices. Such a shortest path tree may have large weight, so we use an idea from \cite{CDG06} and apply an SLT rooted at $K_i$, which balances between the weight and the distortion from $K_i$.
	
	The main issue with the construction described above is that the weight of the spanner in each phase can be proportional to that of the MST, but we have $\log\log n$ of those. Obtaining constant lightness, completely independent of $n$, requires a subtler argument. We use the fact that the weight of the
	light spanners in each phase come "mostly" from the MST, and then some additional weight. We ensure that all the spanners will have {\em the same} MST.
	Then we select the parameters carefully, so that the additional weights will be small enough to form converging sequences, without affecting the distortion by too much.

	\subsection{Related Work}
	
	Partial and scaling embeddings\footnote{A partial embedding (introduced by \cite{KSW04} under the name {\em embedding with slack}) requires that for a {\em fixed} $0<\epsilon<1$, the distortion of all but an $\epsilon$-fraction of the pairs is bounded by the appropriate function of $\epsilon$.} have been studied in several papers \cite{KSW04,ABCD05,ABN11,CDG06,ABN15,ABNS14}. Some of the notable results are embedding arbitrary metrics into a distribution over trees \cite{ABCD05} or into Euclidean space $\cite{ABN11}$ with tight $O(\log(1/\epsilon))$ scaling distortion.
The notion $\ell_q$-distortion was introduced in \cite{ABN11}, they show that their scaling distortion results imply constant average distortion and $O(q)$ bound on the $\ell_q$-distortion. This notion has been further studied in several papers, including \cite{ABN15,ABNS14,CDG06}, and most recently applied in the context of dimensionality reduction~\cite{BFN17}. In \cite{ABN15}, an embedding into a single spanning tree with tight $O(1/\sqrt{\epsilon})$ scaling distortion was shown, which implies, in particular, constant average distortion, but there is no guarantee on the weight of the tree. It follows from \cite{ABCD05} that this bound is tight even when embedding into arbitrary (non-spanning) trees.
	
	Prioritized distortion embeddings were studied in \cite{EFN15}, for instance they give an embedding of arbitrary metrics into a distribution over trees with expected prioritized distortion $O(\log j)$, and into Euclidean space with prioritized distortion $\tilde{O}(\log j)$.
	
	Probabilistic embedding into trees \cite{bartal1,bartal2,bartal3,FRT03} and spanning trees \cite{AKPW95,EEST05,ABN08,AN12} has been intensively studied, and found numerous applications to approximation and online algorithms, and to fast linear system solvers. While our distortion guarantee does not match the best known worst-case bounds, which are $O(\log n)$ for arbitrary trees \cite{bartal3,FRT03} and $\tilde{O}(\log n)$ for spanning trees \cite{ABN08,AN12}, we give the first probabilistic embeddings into spanning trees with polylogarithmic scaling distortion in which all the spanning trees in the support of the distribution are light.

	The paper \cite{CDG06} considers partial and scaling embedding into spanners, and show a general transformation from worst-case distortion to partial and scaling distortion. In particular, they show a spanner with $O(n)$ edges and $O(\log(1/\epsilon))$ scaling distortion. For a fixed $\epsilon>0$, they also obtain a spanner with $O(n)$ edges, $O(\log(1/\epsilon))$ partial distortion and lightness $O(\log(1/\epsilon))$.\footnote{The original paper claims lightness $O(\log^2 (1/\epsilon))$, but their proof in fact gives the improved bound.} Note that these results fall short of achieving both constant average distortion and constant lightness.
	
In a subsequent work, \cite{FS16} used our general reduction for light spanners (\theoremref{thm:reduction_lightness}), to show that the $O(\log n/\rho)$-greedy spanner has lightness $1+\rho$.

	\section{Preliminaries}\label{sec:pre}
	
	All the graphs $G=(V,E,w)$ we consider are undirected and weighted with nonnegative weights. We shall assume w.l.o.g that all edge weights are different. If it is not the case, then one can break ties in an arbitrary (but consistent) way. Note that under this assumption, the MST $T$ of $G$ is unique. The weight of a graph $G$ is $w(G)=\sum_{e\in E}w(e)$.
	Let $d_G$ be the shortest path metric on $G$. For a subset $K\subseteq V$ and $v\in V$ let $d_G(v,K)=\min_{u\in K}\{ d_G(u,v)\}$. For $r\ge 0$ let $B_G(v,r)=\{u\in V: d_G(u,v)\le r\}$ (we often omit the subscript when clear from context).
	
	For a graph $G=(V,E)$ on $n$ vertices, a subgraph $H=(V,E')$ where $E'\subseteq E$ (with the induced weights) is called a \emph{spanner} of $G$. We say that a pair $u,v\in V$ has \emph{distortion} at most $t$ if
	\[
	d_H(v,u)\le t\cdot d_G(v,u)~,
	\]
	(note that always $d_G(v,u)\le  d_H(v,u)$).
	If every pair $u,v\in V$ has distortion at most $t$, we say that the spanner $H$ has distortion $t$.
	Let $T$ be the (unique) MST of $G$, the \emph{lightness} of $H$ is the ratio between the weight of $H$ and the weight of the MST, that is $\Psi(H)=\frac{w(H)}{w(T)}$. We sometimes abuse notation and identify a spanner or a spanning tree with its set of edges.

	A metric space $(X,d_X)$ is defined over a set of points $X$ and a nonnegative distance function $d_X$, with positive values on distinct points, and obeying the triangle inequality. Every weighted graph $G$ can be viewed as a metric space $(V,d_G)$.
	For two metric spaces $(X,d_X)$, $(Y,d_Y)$ and a non-contractive embedding $f:X\to Y$,\footnote{An embedding $f$ is non-contractive if for every $x,y\in X$, $d_Y(f(x),f(y))\ge d_X(x,y)$.} the distortion of a pair $x,y\in X$ under $f$ is defined as $\frac{d_Y(f(x),f(y))}{d_X(x,y)}$.
	
	When considering a graph $G$ and its subgraph $H$, we may view the metric of $G$ as being embedded into $H$ via the identity map, in which case the last definition of distortion given above coincides with the those given earlier. Hence, the following definitions may be interpreted in the graph case in the obvious way.

	\paragraph{Prioritized Distortion.}
	Let $(X,d_X), (Y,d_Y)$ be metric spaces. Let $\pi=v_1,\dots,v_n$ be a priority ranking (an ordering) of the points (vertices) of $X$, and let  $\alpha:\mathbb{N}\rightarrow\mathbb{R}_+$ be some monotone non-decreasing function.
	We say that a non-contractive embedding $f:X\to Y$ has  {\em prioritized distortion} $\alpha$ (w.r.t $\pi$), if for all $1\le j<i\le n$, the pair ${v_j,v_i}$ has distortion  (under $f$) at most $\alpha(j)$.
	
	\paragraph{Scaling Distortion.} Let $(X,d_X), (Y,d_Y)$ be metric spaces, with $|X|=n$. For $v\in X$ and $\epsilon\in\left(0,1\right)$ let $R(v,\epsilon)=\min\left\{ r:|B(v,r)|\ge\epsilon n\right\}$
.
	A point $u \in X$ is called $\epsilon$-{\em far} from $v$ if $d_X(u,v)\ge R(v,\epsilon)$.
	Given a function $\gamma:\left(0,1\right)\rightarrow\mathbb{R}_{+}$, we say that a non-contractive embedding $f:X\to Y$ has \emph{scaling distortion} $\gamma$, if for every $\epsilon\in\left(0,1\right)$, there are at least $(1-\epsilon){|X|\choose 2}$ pairs that have distortion at most $\gamma(\epsilon)$.
	We say that $f$ has \emph{coarse} scaling distortion $\gamma$, if
	every pair $v,u\in X$ such that both $u,v$ are $\epsilon/2$-far from each other, has distortion at most $\gamma(\epsilon)$.
	\footnote{It can be verified that coarse scaling distortion $\gamma$ implies scaling distortion $\gamma$.}

	\paragraph{Moments of Distortion.}
	Let $(X,d_X), (Y,d_Y)$ be metric spaces. For $1 \leq q \leq \infty$, define the \emph{$\ell_q$-distortion} of
	a non-contractive embedding $f:X\to Y$ as:
	$$
	\dist_q(f) = \E\left[\left(\frac{d_{Y}(f(u),f(v))}{d_{G}(u,v)}\right)^{q}\right]^{1/q} ,
	$$
	where the expectation is taken according
	to the uniform distribution over ${X\choose 2}$.
	The
	classic notion of \emph{distortion} is expressed by the
	$\ell_\infty$-distortion and the \emph{average distortion} is
	expressed by the $\ell_1$-distortion. The following was proved in \cite{ABN15}.
	
	\begin{lemma}\label{lemma:scaling-q-norm}(\cite{ABN15})
		Let $(X,d_X), (Y,d_Y)$ be metric spaces. If a non-contractive embedding $f:X\to Y$ has scaling distortion $\gamma$ then
		\begin{equation*}
		\dist _q(f) \leq \left( 2\int_{\frac{1}{2} {n\choose
				2}^{-1}}^{1} \gamma(x)^q dx
		\right)^{1/q} ~ .
		\end{equation*}
	\end{lemma}

		\section{Prioritized Distortion vs. Coarse Scaling Distortion}
	
	In this section we study the relationship between the notions of prioritized and scaling distortion.  We show that there is a reduction that allows to transform embeddings with prioritized distortion into embeddings with coarse scaling distortion, and vice versa.
	
\subsection{Coarse Scaling Distortion implies Prioritized Distortion}
	
	The following theorem shows that coarse scaling distortion implies prioritized distortion, implying some new prioritized distortion embedding results, and in particular a prioritized version of Bourgain's theorem.
	
	\begin{theorem}\label{thm:coarse-scaling to priority}
		Let $\mu:\N\to \mathbb{R}^+$ be a non-increasing function such that $\sum_{i \geq 1} \mu(i) = 1$.
		Let ${\cal Y}$ be a family of finite metric spaces, and assume that for every finite metric space $(Z,d_Z)$ there exists a non-contractive embedding $f_Z:Z\to Y_Z$, where $(Y_Z,d_{Y_Z}) \in {\cal Y}$, with (monotone non-increasing) coarse scaling distortion $\gamma$. Then, given a finite metric space $(X,d_X)$ and a priority ranking $x_1,\dots,x_n$ of the points of $X$, there exists an embedding $f:X\to Y$, for some $(Y,d_Y) \in {\cal Y}$, with prioritized distortion $\gamma(\mu(i))$.
	\end{theorem}
	
	\begin{proof}
		Given the metric space $(X,d_X)$ and a priority ranking $x_1,\dots,x_n$ of the points of $X$, let $\delta = \min_{i \neq j} d_X(x_i,x_j)/2$. We define a new metric space $(Z,d_Z)$ as follows. For every $1 \leq i \leq n$, every point $x_i$ is replaced by a set $X_i$ of $|X_i| = \lceil \mu(i) n \rceil$ points, and let $Z=\bigcup_{i=1}^n X_i$. For every $u \in X_i$ and $v \in X_j$ define $d_Z(u,v) = d_X(x_i,x_j)$ when $i \neq j$, and $d_Z(u,v) = \delta$ otherwise.
		Observe that $|Z| = \sum_{i=1}^n |X_i| \leq \sum_{i=1}^n (\mu(i) n + 1) \le 2n$.
		
		We now use the embedding $f_Z:Z\to Y_Z$ with coarse scaling distortion $\gamma$, to define an embedding $f:X\to Y_Z$, by letting for every $1 \leq i \leq n$, $f(x_i) = f_Z(u_i)$ for some (arbitrary) point $u_i \in X_i$.
		By construction of $Z$, for every $j> i$, we have that $X_i\subseteq B(u_i,d_Z(u_i,u_j))\cap B(u_j,d_Z(u_i,u_j))$. As $|X_i|\ge \mu(i)n\ge  \frac{\mu(i)}{2}|Z|$, it holds that $u_i,u_j$ are $\epsilon/2$-far from each other for $\epsilon=\mu(i)$.
		This implies that $\frac{d_{Y_Z}(f(x_i),f(x_j))}{d_X(x_i,x_j)} = \frac{d_{Y_Z}(f_Z(u_i),f_Z(u_j))}{d_Z(u_i,u_j)} \leq \gamma(\mu(i))$.
	\end{proof}
	
	It follows from a result of \cite{EFN15} that the convergence condition on $\mu$ in the above theorem is necessary (more details below).
	We note that this reduction can also be applied to cases where the coarse scaling embedding is only known for a class of metric spaces (rather than all metrics), as long as the transformation needed for the proof can be made so that the resulting new space is still in the class. This holds for most natural classes, such as metrics arising from trees, planar graph, graphs excluding a fixed minor, bounded degree graphs, doubling metrics, etc.
	\footnote{For the graph classes mention above (as well as for doubling metrics), a small change in the construction is needed. From each original vertex $x_i$, we will grow a path $X_i$ of $\lceil \mu(i) n \rceil$ vertices, where all the path edges have weight $\frac{\delta \cdot \alpha}{2n}$ for arbitrarily small $\alpha>0$. We will also need to choose $u_i$ to be the single leaf in the added path (rather then simply arbitrarily chosen vertex). The same proof will guarantee a $(1+\alpha)\gamma(\mu(i))$ prioritized distortion.}

\subsubsection{Implications}

	The reduction implies that all existing {\em coarse} scaling distortion embeddings and distance oracles
have priority distortion counterparts, thus improving few of the results of \cite{EFN15}.
\paragraph{Embeddings.} By applying a theorem of \cite{ABN11} we get the following.
	
	\begin{corollary} For every $1 \leq p \leq \infty$ and every finite metric space $(X,d_X)$ and a priority ranking of $X$, there exists an embedding with prioritized distortion $O(\log j)$ into $l_p^{O(\log |X|)}$.
	\end{corollary}
	
	Another consequence of the results of \cite{ABN11} is better bounds for decomposable metrics\footnote{Roughly, a metric space is called $\tau$-decomposable if it allows probabilistic partitions with padding parameter $\tau$; e.g. Planar metrics and doubling metrics. An exact definition appears in \cite{ABN11}.}:
	
		\begin{corollary} For every $1 \leq p \leq \infty$ and every finite $\tau$-decomposable metric space $(X,d_X)$ and a priority ranking of $X$, there exists an embedding with prioritized distortion $O(\tau^{1-1/p}(\log j)^{1/p})$ into $l_p^{O(\log^2 |X|)}$.
	\end{corollary}

\paragraph{Spanners.}	
	Applying \theoremref{thm:coarse-scaling to priority} on \cite[Corollary 3]{CDG06} we get a linear size prioritized spanner.
	\begin{corollary}
		Given a graph $G=\left(V,E\right)$ and any priority ranking $v_{1},v_{2},\dots,v_{n}$ of $V$, there exists a spanner $H$
		with $O(n)$ edges and prioritized distortion $O\left(\log j\right)$.
	\end{corollary}

We remark that in \theoremref{thm:Prioritized-Spanner} we show directly a spanner with $O(n\log\log n)$ edges and prioritized distortion $\tilde{O}(\log j)$ (which could easily be made $O(\log j)$). While not being of linear size, that spanner is very light. We currently do not know how to achieve both lightness and linear size spanner with prioritized distortion $O(\log j)$.

\paragraph{Distance Oracles.} In \cite{EFN15}, among other possible tradeoffs, it was shown how to construct distance oracles with $O(n\log\log n)$ space and prioritized distortion $O(\log n/\log(n/j))$ with $O(1)$ query time. (Alternatively, they had $O(n\log^*n)$ space and $O(2^{\log n/\log(n/j)})$ prioritized stretch with $O(1)$ query time.) The space requirement in \cite{EFN15} was never truly linear in $n$.	
Chechick \cite{C15} showed that for every metric space $(X,d_X)$, one can construct a distance oracle with $O(\log n)$-stretch, $O(1)$-query time and $O(n)$ space. A black box reduction from \cite{CDG06}, will provide us with distance oracle with $O(\log\frac1\epsilon)$ coarse scaling distortion, $O(1)$-query time and $O(n)$ space. We conclude with a linear size prioritized distance oracle.
	\begin{corollary}
		For every metric space  $(X,d_x)$ and every priority ranking, there exist a distance oracle requiring $O(n)$ space, that answer distance queries in $O(1)$ time and $O(\log j)$ priority distortion.		
	\end{corollary}
We remark that the prioritized distortion $O(\log n/\log(n/j))$ of \cite{EFN15} is superior to our $O(\log j)$.

	\subsection{Prioritized Distortion implies Coarse Scaling Distortion}		
	
	Here we prove the direction that is used for our main result of a light constant average distortion spanning tree, specifically, that prioritized distortion implies scaling distortion.

	\begin{theorem}\label{thm:priority to scaling}
		Let $(X,d_X)$, $(Y,d_Y)$ be metric spaces, then there exists a priority ranking $\pi=x_1,\dots,x_n$ of the points of $X$ such that the following holds: If there exists an embedding $f:X\to Y$ with (monotone non-decreasing) prioritized distortion $\alpha$ (with respect to $\pi$), then $f$ has coarse scaling distortion $O(\alpha(8/\epsilon))$.
	\end{theorem}
	
	The basic idea of the proof is to choose the priorities so that for every $\epsilon$, every $v\in X$ has a representative $v'$ of sufficiently high priority within distance $\approx R(v,\epsilon)$. Then for any $u\in X$ which is $\epsilon$-far from $v$, we can use the low distortion guarantee of $v'$ with both $v$ and $u$ via the triangle inequality. To this end, we employ the notion of a {\em density net} due to \cite{CDG06}, who showed that a greedy construction provides such a net.
	
	\begin{definition}[Density Net]\label{def:den}
		Given a metric space $(X,d)$
		and a parameter $0<\epsilon<1$, an $\epsilon$\emph{-density-net} is a set $N\subseteq X$ such that: \textbf{1)} for all $v\in X$ there exists $u\in N$ with $d(v,u)\le 2R(v,\epsilon)$ and \textbf{2)} $\left|N\right|\le\frac{1}{\epsilon}$.
	\end{definition}

	\begin{proof}[Proof of \theoremref{thm:priority to scaling}]
		We begin by describing $\pi$, the desired priority ranking of $X$. For every integer $1\le i\le\lceil\log n\rceil$ let $\epsilon_{i}=2^{-i}$, and let $N_i\subseteq X$
		be an $\epsilon_i$-density-net in $X$. Set $\pi$ to be a priority ranking of $X$ satisfying that every point $v\in N_i$ has priority at most $\left|\bigcup_{j=1}^{i}N_j\right|\le\sum_{j=1}^{i}\left|N_j\right|$. As for any $j$,
		$\left|N_j\right|\le\frac{1}{\epsilon_{j}}=2^j$, each point in $N_i$ has priority at most $\sum_{j=1}^{i}\frac{1}{\epsilon_{j}}\le\sum_{j=1}^{i}2^j<2^{i+1}$.
		
		Let $f:X\rightarrow Y$ be some non-contractive embedding with priority distortion $\alpha$ with respect to $\pi$. Fix some
		$\epsilon\in(0,1)$ and a pair
		$v,u\in V$ so that $u$ is $\epsilon$-far from $v$.
		Let $i$ be the minimal integer such that
		$\epsilon_{i}\le\epsilon$ (note that we may assume $1\le i\le\lceil\log n\rceil$, because there is nothing to prove for $\epsilon<1/n$). By \defref{def:den} we can take $v'\in N_i$ such that $d(v,v')\le 2R(v,\epsilon_i)$.
		As $u$ is $\epsilon$-far from $v$, it holds that
		\begin{equation}\label{eq:dd}
		d_{X}(v,v')\le 2R\left(v,\epsilon_{i}\right)\le2R\left(v,\epsilon\right)\le2 d_{X}(v,u)~.
		\end{equation}
		In particular, by the triangle inequality,
		\begin{equation}\label{eq:ss}
		d_{X}(u,v')\le d_{X}(u,v)+d_{X}(v,v')\stackrel{\eqref{eq:dd}}{\le} 3d_{X}(u,v)\ .
		\end{equation}
		The priority of $v'$ is at most $2^{i+1}$, hence  		
		\begin{eqnarray*}
			\lefteqn{d_{Y}(f(v),f(u))}\\
			&\le&d_{Y}(f(v),f(v'))+d_{Y}(f(v'),f(u))\\
			&\le&\alpha(2^{i+1})\cdot d_{X}(v,v')+\alpha(2^{i+1})\cdot d_{X}(v',u)\\
			&\stackrel{\eqref{eq:dd}\wedge\eqref{eq:ss}}{\le}&5\alpha(2/\epsilon_{i})\cdot d_{X}(v,u)~.
		\end{eqnarray*}
		
		By the minimality of $i$ it follows that $1/\epsilon_i\le 2/\epsilon$, and since $\alpha$ is monotone
		
		\[
		d_{Y}(f(v),f(u))\leq5\alpha(2/\epsilon_{i})\cdot d_{X}(v,u)\le5\alpha(4/\epsilon)\cdot d_{X}(v,u)\ ,
		\]
		as required. Since we desire distortion guarantee for pairs that are $\epsilon/2$-far, the distortion becomes $O(\alpha(8/\epsilon))$.

	\end{proof}

\begin{remark}
		The proof of \theoremref{thm:priority to scaling} provides an even stronger conclusion, that any pair $u,v\in X$ such that one is $\epsilon/2$-far from the other, has the claimed distortion bound. While in the original definition of coarse scaling both points are required to be $\epsilon/2$-far from each other, it is often the case that we achieve the stronger property. Yet, in some of the cases in previous work the weaker definition seemed to be of importance. Combining \theoremref{thm:priority to scaling} and \theoremref{thm:coarse-scaling to priority}, we infer that essentially any coarse scaling embedding can have such a one-sided guarantee, with a slightly worse dependence on $\epsilon$, as claimed in the following corollary.
	\end{remark}

	\begin{corollary}\label{cor:weak_to_strong_scaling}
		Let $\mu:\N\to \mathbb{R}^+$ be a non-increasing function such that $\sum_{i \geq 1} \mu(i) = 1$.
		Let ${\cal Y}$ be a family of finite metric spaces, and assume that for every finite metric space $(Z,d_Z)$ there exists a non-contractive embedding $f_Z:Z\to Y_Z$, where $(Y_Z,d_{Y_Z}) \in {\cal Y}$, with (monotone non-increasing) coarse scaling distortion $\gamma(\epsilon)$. Then given any finite metric space $X$, there exists an embedding $f:X\to Y$, for some $(Y,d_Y) \in {\cal Y}$, with (monotone non-decreasing) {\em one-sided} coarse scaling distortion $O(\gamma(\mu(8/\epsilon)))$.
	\end{corollary}
	\begin{proof}
		By the assumption, there exists $(Y,d_Y)\in {\cal Y}$ so that $X$ embeds to $Y$ with coarse scaling distortion $\gamma(\epsilon)$.
		According to \theoremref{thm:coarse-scaling to priority}, there is an embedding $f$ with prioritized distortion $\gamma(\mu(i)))$ (w.r.t to any fixed priority ranking $\pi$).
		We pick $\pi$ to be the ordering required by \theoremref{thm:priority to scaling}, and conclude that $f$ has one-sided coarse scaling distortion $O(\gamma(\mu(8/\epsilon)))$.
	\end{proof}

\subsection{Coarse Partial Distortion and Terminal Distortion}\label{sec:terminal-partial}

	As a special case of the reductions \theoremref{thm:coarse-scaling to priority} and \theoremref{thm:priority to scaling}, we can prove an equivalence between coarse partial distortion to terminal distortion.
	\begin{definition}[Coarse partial distortion]
		Let $(X,d_X), (Y,d_Y)$ be metric spaces, and let $\epsilon\in(0,1)$, $\gamma\ge 1$.
		A non-contractive embedding $f:X\to Y$ has $(1-\epsilon)$-coarse partial scaling distortion $\gamma$,  if every pair $v,u\in X$ such that both $u,v$ are $\epsilon/2$-far from each other, has distortion at most $\gamma$.			
	\end{definition}
	Note that the embedding $f$ has coarse scaling distortion $\gamma$ if and only if for every $\epsilon\in(0,1)$, $f$ has $(1-\epsilon)$-coarse partial distortion $\gamma(\epsilon)$.
	\begin{definition}[Terminal distortion]\label{def:terminalDist}
		Let $(X,d_X), (Y,d_Y)$ be metric spaces, and $K\subseteq X$ a subset of terminals.
		A non-contractive embedding $f:X\to Y$ has terminal distortion $\alpha$ w.r.t. $K$,  if every pair $(v,u)\in K\times X$ has distortion at most $\alpha$.			
	\end{definition}
	Note that for a priority ranking $\pi=x_1,\dots,x_n$, the embedding $f$ has priority distortion $\alpha$ w.r.t $\pi$ if and only if for every $k$, $f$ has terminal distortion $\alpha(k)$ w.r.t. $K=\{x_1,\dots,x_k\}$.
	It is important to note that \defref{def:terminalDist} differs from the original definition of terminal distortion in \cite{EFN15T}, which did not require $f$ to be non-contractive on {\em all pairs}. We elaborate on this issue in \subsectionref{sec:weak}.

	\begin{theorem}\label{thm:coarse-partial to terminal}
	Let $\mu:\N\to \mathbb{R}^+$ be a non-increasing function such that $\sum_{i \geq 1} \mu(i) = 1$. Let $k \in \mathbb{N}$.
	Let ${\cal Y}$ be a family of finite metric spaces, and assume that for every finite metric space $(Z,d_Z)$ there exists an embedding $\phi:Z\to Y_Z$, where $(Y_Z,d_{Y_Z}) \in {\cal Y}$, with coarse $(1-1/(2k))$-partial distortion $\gamma$. Then, given a finite metric space $(X,d_X)$ and a set of terminals $K\subsetneq X$ of size $|K|=k$, there exists an embedding $f:X\to Y$, for some $(Y,d_Y) \in {\cal Y}$, with terminal distortion $\gamma$.
\end{theorem}
\begin{proof}
	Simply follow the proof of \theoremref{thm:coarse-scaling to priority}, using $\mu(x) = \frac{1}{2k}$ for $x\in K$, and $\mu(x) = \frac{1}{2(|X|-k)}$ for $x\in X\setminus K$.
	As every $x\in K$ has $\frac{n}{2k}$ copies, and the new metric $Z$ contains at most $2n$ points,
	$x$ is $\frac{1}{4k}$-far from any other $y\in X$ (in the metric space $Z$). Also this $y$ is $\frac{1}{4k}$-far from $x$ (since $|B_Z(y,d(x,y))|\ge n/(2k)$). Thus
the embedding with coarse $(1-1/(2k))$-partial distortion for $Z$ has distortion at most $\gamma$ for such a pair $x,y$.
\end{proof}

	\begin{theorem}\label{thm:terminal to coarse-partial}
		Let $(X,d_X)$, $(Y,d_Y)$ be metric spaces, and $k\in\mathbb{N}$ a parameter. There exists a subset $K\subseteq X$ of size $k$, such that the following holds: If there exists an embedding $f:X\to Y$ with terminal distortion $\alpha$, then $f$ has $(1-\frac2k)$-coarse scaling distortion $5\alpha$.
	\end{theorem}
	\begin{proof}
		Following the lines of the proof of \theoremref{thm:priority to scaling} let $K$ be a $\frac1k$-density net.
		Fix a pair $v,u\in V$ so that $u$ is $\frac12\cdot\frac2k=\frac1k$-far from $v$.
		Let $v'\in N$ such that $d_X(v,v')\le 2R(v,\frac1k)\le2  d_{X}(v,u)$. It holds that,
		\begin{align*}
		d_{Y}(f(v),f(u)) & \le d_{Y}(f(v),f(v'))+d_{Y}(f(v'),f(u))\\
		& \le\alpha\cdot d_{X}(v,v')+\alpha\cdot d_{X}(v',u)\\
		& \le2\alpha\cdot d_{X}(v,u)+3\alpha\cdot d_{X}(v,u)=5\alpha\cdot d_{X}(v,u)~.\qedhere
		\end{align*}
		
	\end{proof}

	Among other implications, \theoremref{thm:terminal to coarse-partial} implies the following:
	\begin{corollary}For every parameters $0<\epsilon,\delta<1$, every $n$-vertex weighted graph $G$ contains a spanning tree with $(1-\epsilon)$-coarse partial distortion $O(\frac{1}{\epsilon\delta})$ and $1+\delta$ lightness.
	\end{corollary}
\begin{proof}
In \cite{EFN15T} it was shown that for every weighted graph $G=(V,E,w)$ and terminal set $K\subseteq V$ of size $k$, there is a spanning tree $T$ with terminal distortion $O(k)$ and constant lightness.
Using \theoremref{thm:reduction_lightness} (proven below), we get that $G$ contains a spanning tree with terminal distortion $O(k/\delta)$ and  lightness $1+\delta$.
Now, \theoremref{thm:terminal to coarse-partial} (with $k=\frac2\epsilon$) implies the corollary.
\end{proof}

\subsubsection{Weak Terminal Distortion}\label{sec:weak}

Our definition of terminal distortion has a one-sided guarantee on {\em all pairs}, e.g., the embedding must not contract any distance.
This definition differs from the original definition of terminal distortion which appears in \cite{EFN15T}, where the non-contractive requirement was missing (formally, by \cite{EFN15T}, $f$ has terminal distortion $\alpha$ iff there is some constant $c\in \mathbb{R}$ such that $\forall(v,u)\in K\times X$, $d_X(u,v)\le c\cdot d_Y(f(u),f(v)) \le \alpha\cdot d_X(u,v)$ ). We will refer to the original definition from \cite{EFN15T} as \emph{weak terminal distortion}.

This two definitions are indeed different. For example, in \cite{EFN15T} it was shown that given $n$ points containing $k$ terminals in $\mathbb{R}^n$, they can be embedded into $\mathbb{R}^{O(\log k)}$ with weak terminal distortion $O(1)$ (under the $\ell_2$-norm). However, any non-contracting embedding with constant distortion requires $\Omega(\log n)$ dimensions, so this is impossible under our \defref{def:terminalDist}.
As a result of the difference between these definitions, there are some results in \cite{EFN15T} on which the reduction of \theoremref{thm:terminal to coarse-partial} cannot be used.	

Nevertheless, if the target space is $\ell_p$, we devise a transformation from weak terminal distortion into terminal distortion, while increasing the dimension {\em additively} by $O(\log n)$. The first step is \theoremref{thm:generalToTerminal}, in which we extend a standard embedding into a terminal one, in a different manner than \cite{EFN15T}. This theorem has other implications: in particular, we generalize and improve the dimension in a result of  \cite{ABCD05,ABN11} on embedding into $\ell_p$ with coarse partial distortion.
\begin{theorem}\label{thm:generalToTerminal}
	Let $(X,d_X)$ be metric space of size $n$, and $K\subseteq X$ be a subset of terminals. Suppose that there exists an embedding $f:K\to \ell_p^\beta$ with distortion $\alpha$, then there is an embedding $\hat{f}:X\rightarrow\ell_p^{\beta+O(\log n)}$ with terminal distortion $O(\alpha)$.
\end{theorem}

\begin{proof}
	Assume, as we may, that $f$ is non-contractive. That is, for every $v,u\in K$, $d_X(u,v)\le \Vert f(u)-f(v)\Vert_p \le \alpha\cdot d_X(u,v)$.
	
	Fix $m=O(\log n)$. Let $g:X\rightarrow\{\pm 1\}^m$ such that for every $v,u\in X$, there are at least $\frac m4$ coordinates $i$ where $g_i(v)\ne g_i(u)$ (a random $g$ will work with high probability, as can be verified by Chernoff inequality).
	For every vertex $u\in X$, let $k_u$ be the closest terminal to $u$.
	The embedding $\hat{f}$ is defined as follows. For $u\in X$,
	$$\hat{f}(u)=f(k_{u})\oplus\frac{d_X(u,k_u)}{m^{1/p}}\cdot g(u)~.$$
	
	First, we will show that $\hat{f}$ has expansion at most $O(\alpha)$ on terminal pairs. Fix some $v\in K$ and $u\in X$.
	\begin{align*}
	\Vert\hat{f}(v)-\hat{f}(u)\Vert_{p}^{p} & =\Vert f(v)-f(k_{u})\Vert_{p}^{p}+\frac{1}{m}\sum_{i=1}^{m}\left|g_{i}(u)\cdot d_{X}(u,k_{u})\right|^{p}\\
	& \le\alpha^{p}\cdot d_{X}\left(v,k_{u}\right)^{p}+d_{X}\left(u,k_{u}\right)^{p}\\
	& \le\left(\alpha^{p}+1\right)\cdot\left(d_{X}\left(v,u\right)+d_{X}\left(u,k_{u}\right)\right)^{p}\\
	& \le\left(\alpha^{p}+1\right)\cdot\left(2d_{X}\left(v,u\right)\right)^{p}~.
	\end{align*}
	Thus, $\Vert\hat{f}(v)-\hat{f}(u)\Vert_{p}\le2\left(\alpha^{p}+1\right)^{1/p}\cdot d_{X}\left(v,u\right)\le2\left(\alpha+1\right)\cdot d_{X}\left(v,u\right)$.
	
	Next, we bound the contraction for all pairs. Fix some $v,u\in X$. If $d_X(u,v)/2\ge d_X(v,k_v)+d_X(u,k_u)$, then
	\begin{eqnarray*}
		\Vert\hat{f}(v)-\hat{f}(u)\Vert_{p} & \ge & \Vert f(k_{v})-f(k_{u})\Vert_{p}\ge d_{X}(k_{v},k_{u})\\
		& \ge & d_{X}\left(v,u\right)-d_{X}\left(v,k_{v}\right)-d_{X}\left(u,k_{u}\right)\ge d_{X}\left(v,u\right)/2~.
	\end{eqnarray*}
	Otherwise,
	\begin{align*}
	\Vert\hat{f}(v)-\hat{f}(u)\Vert_{p}^{p} & \ge\frac{1}{m}\sum_{i=1}^{m}\left|g_{i}(v)\cdot d_{X}(v,k_{v})-g_{i}(u)\cdot d_{X}(u,k_{u})\right|^{p}\\
	& =\frac{1}{m}\cdot\frac{m}{4}\cdot\left|d_{X}(v,k_{v})+d_{X}(u,k_{u})\right|^{p}\ge\frac{1}{4}\cdot\left(d_{X}(v,u)/2\right)^{p}~.
	\end{align*}
	To ensure the embedding does not contract, our final embedding will be $2^{1+\frac2p}\cdot\hat{f}$.		
\end{proof}

We now show the transformation from weak terminal distortion to (non-contracting) terminal distortion.
(By \theoremref{thm:terminal to coarse-partial} this can provide embeddings with coarse partial distortion as well.)
Suppose an embedding $f:X\rightarrow Y$ has weak terminal distortion $\alpha$. In particular, its restriction to $K$ has distortion $\alpha$. Using \theoremref{thm:generalToTerminal} we conclude:
\begin{corollary}\label{cor:weakToterminal}
	Let $(X,d_X)$ be metric space of size $n$, and $K\subseteq X$ be a subset of terminals. Suppose that there exists an embedding $f:X\to \ell_p^\beta$ with weak terminal distortion $\alpha$, then there exist embedding $\hat{f}:X\rightarrow\ell_p^{\beta+O(\log n)}$ with terminal distortion $O(\alpha)$.
\end{corollary}	

Fix some $p\ge 1$. Let $\mathcal{X}$ be a subset-closed family of finite metric spaces such that for any $n\ge 1$ and any $n$-point metric space $X\in\mathcal{X}$ there exists an embedding $f_X:X\rightarrow\ell_p$ with distortion $\alpha(n)$ and dimension $\beta(n)$.
In \cite{ABCD05,ABN11} it was shown that, assuming all $f_X$ are strongly non-expansive,
\footnote{Embedding $f:X\rightarrow\ell_p$ is strongly non-expansive if $f=(\eta_1f_1,\dots,\eta_mf_m)$ where $\sum_{i=1}^{m}\eta_i=1$, and each $f_i$ is non-expansive embedding into $\mathbb{R}$.}
there is a universal constant $C$ and an embedding from $X$ into $\ell_p$ with $(1-\epsilon)$-coarse partial distortion $O(\alpha(C/\epsilon))$ and dimension $\beta(C/\epsilon)\cdot O(\log n)$.
By combining \theoremref{thm:generalToTerminal} with \theoremref{thm:terminal to coarse-partial} we  considerably improve the dimension, and remove the strongly non-expansive requirement.
\begin{corollary}\label{cor:GeneralToPartial}
	Fix some $p\ge 1$. Let $\mathcal{X}$ be a subset-closed family of finite metric spaces such that for any $n\ge 1$ and any $n$-point metric space $X\in\mathcal{X}$ there exists an embedding $f_X:X\rightarrow\ell_p$ with distortion $\alpha(n)$ and dimension $\beta(n)$.
	Then there is an embedding from $X$ into $\ell_p$ with $(1-\epsilon)$-coarse partial distortion $O(\alpha(2/\epsilon))$ and dimension $\beta(2/\epsilon)+O(\log n)$.
\end{corollary}
	
	\section{Light Spanner with Prioritized Distortion}
	
	In this section we prove that every graph admits a light spanner with bounded prioritized distortion.
	\begin{theorem}[Prioritized Spanner]\label{thm:Prioritized-Spanner}
		Given a graph $G=\left(V,E\right)$, a parameter $0<\rho<1$ and any priority ranking $v_{1},v_{2},\dots,v_{n}$ of $V$, there exists a spanner $H$
		with lightness $1+\rho$ and prioritized distortion $\tilde{O}\left(\log j\right)/\rho$.
	\end{theorem}

	Combining \theoremref{thm:Prioritized-Spanner} and \theoremref{thm:priority to scaling} we obtain the following.
	\begin{theorem}
		\label{cor:light spanner scaling distortion}
		For any parameter $0<\rho<1$, any graph contains a spanner with coarse scaling distortion $\tilde{O}\left(\log\left(1/\epsilon\right)\right)/\rho$ and lightness $1+\rho$.\end{theorem}
	
		 By \lemmaref{lemma:scaling-q-norm} it follows that this spanner has $\ell_q$-distortion $\tilde{O}(q)/\rho$ for any $1\le q<\infty$.
		
		 We can also obtain a spanner with both scaling distortion and prioritized distortion simultaneously, where the priority is with respect to an arbitrary ranking $\pi=v_1,\dots,v_n$. To achieve this, one may define a ranking which interleaves $\pi$ with the ranking generated in the proof of \theoremref{thm:priority to scaling}.  \\

	We now turn to proving \theoremref{thm:Prioritized-Spanner}. The proof is based on the following main technical lemma:
	\begin{lemma}\label{lem:terminal spanner}
		Given a graph $G=\left(V,E\right)$, a subset $K\subseteq V$ of size $k$, and a parameter $0<\delta<1$, there exists a spanner $H$ that \textbf{1)} contains the MST of $G$, \textbf{2)} has lightness $1+\delta$, and \textbf{3)} every pair in $K\times V$ has distortion $O((\log k)/\delta)$.
	\end{lemma}
	
	Before proving this lemma, let us first apply it to prove \theoremref{thm:Prioritized-Spanner}.
	
	\begin{proof}[Proof of \theoremref{thm:Prioritized-Spanner}]
		For every $1\le i\le\lceil\log\log n\rceil$ let $K_{i}=\left\{ v_{j}\ :\ j\le2^{2^{i}}\right\} $. Let $H_{i}$ be the spanner given by \lemmaref{lem:terminal spanner} with respect to the	set $K_{i}$ and the parameter $\delta_i=\rho/i^2$.
		Hence $H_{i}$ has $1+\rho/i^2$ lightness and $O\left(\frac{\log|K_{i}|}{\delta_{i}}\right)=O(2^{i}\cdot i/\rho)$ distortion for pairs in $K_i\times V$. Let $H=\bigcup_{i}H_{i}$ be the union of all these spanners (that is, the graph containing every edge of every one of these spanners). As each $H_{i}$ contains the unique MST of $G$, it holds that
		\begin{eqnarray*}
			\Psi(H) & \le & 1+\sum_{i\ge1}\rho/i^2=1+O\left(\rho\right)\ .
		\end{eqnarray*}

		To see the prioritized distortion, let $v_{j},v_{r}\in V$ be such that $j<r$, and
		let $1\le i\le \lceil\log\log n\rceil$ be the minimal index such that $v_{j}\in K_{i}$. Note that
		$2^{2^{i-1}}\le j$, and in particular $2^{i-1}\le\log j$ (with the exception of $j=1$, but that case holds by the virtue of $j=2$, say). This implies that
		\begin{eqnarray*}
			d_{H}(v_{j},v_{r})&\le& d_{H_{i}}(v_{j},v_{r}) ~\le~ O(2^{i}\cdot i^2/\rho)\cdot d_{G}(v_{j},v_{r})\\
			&\le&\tilde{O}\left(\log j\right)/\rho \cdot d_{G}(v_{j},v_{r})\ .
		\end{eqnarray*}
		as required.
	\end{proof}

	\subsection{Proof of \lemmaref{lem:terminal spanner}}
	The construction of the spanner that fulfills the properties promised in \lemmaref{lem:terminal spanner} is as follows. First, we use the spanner of \cite{CW16} to get a spanner with lightness $O(1)$ and distortion $O(\log k)$ over pairs in $K\times K$.
	Then, by combining this spanner with the SLT by \cite{KRY93}, we expand the $O(\log k)$ distortion guarantee to all pairs in $K\times V$, while the lightness is still $O(1)$. Finally, we use a general reduction (\theoremref{thm:reduction_lightness}), that reduces the weight of a spanner while increasing its distortion. By applying the reduction, we get a spanner with $1+\rho$ lightness while paying additional factor of $1/\rho$ in the distortion.

	We begin by describing the general reduction.
	\begin{theorem}\label{thm:reduction_lightness}
		Let $G=(V,E)$ be a graph, $0<\delta<1$ a parameter and $t:{V\choose 2}\rightarrow\mathbb{R}_{+}$ some function. Suppose that for every weight function $w:E\rightarrow\mathbb{R}_{+}$ there exists a spanner $H$ with lightness $\ell$ such that every pair $u,v\in V$ suffers distortion at most $t(u,v)$.
		Then for every weight function $w$ there exists a spanner $H$ with lightness $1+\delta\ell$  and such that every pair $u,v$ suffers distortion at most $t(u,v)/\delta$.
		Moreover, $H$ contains the MST of $G$ with respect to $w$.

	\end{theorem}

	\begin{proof}
Fix some weight function $w$ and let $G=(V,E,w)$ be the graph associated with this weight function, and
		let $T$ be the MST of $G$. Set $w':E\rightarrow\mathbb{R}_{+}$ to be a new weight function
		$$w'(e)=\begin{cases}
		w(e) & e\in T\\
		w(e)/\delta & e\notin T
		\end{cases}~,$$
		that is, we multiply the weight of all non-MST edges by $1/\delta$. Let
		$G'=(V,E,w')$ be the graph $G$ associated with the new weight function $w'$.
		Note that $T$ is also the MST of $G'$ (since the weight of any spanning tree is higher in $G'$ than in $G$ except for $T$ itself).
		By our assumption there exists a spanner   $H'=\left(V,E_{H'},w'\right)$
		of $G'$ with distortion bounded by $t$ and lightness $\ell$.
		Set $H=\left(V,E_{H'}\cup T,w\right)$ as a spanner
		of $G$. The edge set of $H$ consists of the edges of $H'$ together with the MST edges, all with the original weight function $w$.
		
		As the weight of the non-MST edges are larger in $G'$ by $1/\delta$ factor compared to their weight in $G$, we have
		\begin{eqnarray*}
			w(E_{H}) & = & w(T)+w\left(E_{H'}\setminus T\right) = w(T)+\delta \cdot w'\left(E_{H'}\setminus T\right) \le w(T)+\delta \cdot w'\left(E_{H'}\right)\\
			& \le & w(T)+\delta\ell\cdot w'\left(T\right)=\left(1+\delta\ell\right)\cdot w(T)\ ,
		\end{eqnarray*}
		concluding that the lightness of $H$ is at most $1+\delta\ell$.
		
		To bound the distortion, consider an arbitrary pair of vertices $u,v\in V$.
		Let $P_{u,v}$ be the shortest path from $u$ to $v$ in $G$. As
		for each edge $e\in P_{u,v}$, $w'(e)\le w(e)/\delta$ we have that
        \[
			d_{G'}\left(u,v\right)\le \sum_{e\in P_{u,v}}w'\left(e\right)\le \sum_{e\in P_{u,v}}\frac{1}{\delta}\cdot w\left(e\right) = \frac{1}{\delta}\cdot d_{G}\left(u,v\right)~,
		\]
        Therefore:
		\[
			d_{H}\left(u,v\right)\le d_{H'}\left(u,v\right)\le t(u,v) \cdot d_{G'}\left(u,v\right)\le\frac{t(u,v)}{\delta}d_{G}\left(u,v\right)~,
		\]
		as required.
		\end{proof}

In a recent work,  Chechik and Wulff{-}Nilsen achieved the following result:
	\begin{theorem}[\cite{CW16}]\label{thm:CW16}
		For every weighted graph $G=(V,E,w)$ and parameters $k\ge 1$ and $0<\epsilon\le 1$, there exist a polynomial time algorithm that constructs a spanner with distortion $(2t-1)(1+\epsilon)$ and lightness $n^{1/t}\cdot \poly(\frac{1}{\epsilon})$.
	\end{theorem}
	Note that for an $n$-vertex graph with parameters $t=\log n, \epsilon=1$, they get a spanner with distortion $O(\log n)$ and constant lightness. However, their construction does not seem to provide lightness arbitrarily close to $1$.

	A tree $\mathcal{T}=\left(V',E',w'\right)$ is called a $Steiner\mbox{ }tree$
	for a graph $G=\left(V,E,w\right)$ if $\left(1\right)$ $V\subseteq V'$,
	and $\left(2\right)$ for any pair of vertices $u,v\in V$ it holds
	that $d_{\mathcal{T}}\left(u,v\right)\ge d_{G}\left(u,v\right)$. The $minimum\ Steiner\ tree$
	$T$ of $G$, denoted $SMT\left(G\right)$, is a Steiner tree of $G$
	with minimum weight. It is well-known that for any graph $G$, $w\left(SMT\left(G\right)\right)\ge\frac{1}{2}MST\left(G\right)$.
	(See, e.g., \cite{GilbertP68}, Section 10.)
	
	We will use \cite{CW16} spanner to construct a spanner with $O(1)$ lightness and distortion $O(\log k)$ over pairs in $K\times K$.
	Let $G_k=(K,{K\choose 2},w_k)$ be the complete graph over the terminal set $K$ with weights $w_k(u,v)=d_G(u,v)$ (for $u,v\in K$) that are given by the shortest path metric in $G$.
	Let $T_k$ be the MST of $G_k$. Note that the MST $T$ of $G$ is a Steiner tree of $G_k$, hence $w_k(T_k)\le  2\cdot w_k(SMT(G_k)) \le  2\cdot w(T)$ .
	
	Using \theoremref{thm:CW16}, let $H_k=(K,E_k,w_k)$ be a spanner of $G_k$ with weight $O(w_k(T_k))=O(w(T))$ (constant lightness) and distortion $O(\log k)$. For a pair of vertices $u,v\in K$, let $P_{uv}$ denote the shortest path between $u$ and $v$ in $G$. Let $H'=(V,E',w)$ be a subgraph of $G$ with the set of edges $E'=\cup_{\{u,v\}\in E_k}P_{uv}$ (i.e. for every edge $\{u,v\}$ in $H_k$, we take the shortest path from $u$ to $v$ in $G$).
	It holds that,
	\[
	w(H')\le\sum_{\left\{ u,v\right\} \in E_{k}}w(P_{uv})=\sum_{e\in E_{k}}w_{k}(e)=O(w(T))~.
	\]
	Moreover, for every pair $u,v\in K$,
	\begin{equation}\label{eq:4rr}
			d_{H'}(u,v)\le d_{H_k}(u,v)\le O(\log k)\cdot d_{G_k}(u,v)=O(\log k)\cdot d_G(u,v)~.
	\end{equation}

	Now we extend $H'$ so that every pair in $K\times V$ will suffer distortion at most $O(\log k)$. To this end, we use the following lemma regarding shallow light trees (SLT), which is implicitly proved in \cite{KRY93,ABP92}.
	
	\begin{lemma}\label{lem:slt_term}
		\label{lem: slt multiple vertices }Given a graph $G=\left(V,E\right)$, a parameter $\alpha>1$, and a subset $K\subseteq V$, there exists a spanner $S$ \footnote{In fact, $S$ is a spanning forest of $G$.} of $G$ with lightness $1+\frac{2}{\alpha-1}$, and  for any vertex $u\in V$, $d_{S}(u,K)\le\alpha\cdot d_{G}(u,K)$. 
	\end{lemma}

	Let $S$ be the spanner of \lemmaref{lem: slt multiple vertices } with respect to the set $K$ and parameter $\alpha=2$. Define $H"$ as the union of $H'$ and $S$. As both $H'$ and $S$ have constant lightness, so does $H"$. It remains to bound the distortion of an arbitrary pair $v\in K$ and $u\in V$. Let $k_{u}\in K$
	be the closest vertex to $u$ among the vertices in $K$ with respect
	to the distances in the spanner $S$. By the assertion of \lemmaref{lem: slt multiple vertices },
	\begin{equation}\label{eq:re}
	d_{S}(u,k_{u})=d_{S}(u,K)\le2\cdot d_{G}(u,K)\le 2\cdot d_{G}(u,v)~.
	\end{equation}

	Using the triangle inequality,

	\begin{equation}\label{eq:rr}
d_{G}(v,k_{u})\le d_{G}(v,u)+d_{G}(u,k_{u})\le d_{G}(v,u)+d_{S}(u,k_{u})\stackrel{\eqref{eq:re}}{\le}3\cdot d_{G}(v,u)~.
	\end{equation}
	Since both $v,k_u\in K$ it follows that
	\begin{equation}\label{eq:ee}
d_{H'}(v,k_{u})\stackrel{\eqref{eq:4rr}}{\le}O(\log k)\cdot d_{G}(v,k_{u})\stackrel{\eqref{eq:rr}}{\le}O(\log k)\cdot d_{G}(v,u)~.
	\end{equation}	
	We conclude that
	\begin{equation*}
d_{H"}\left(v,u\right)\le d_{H'}\left(v,k_{u}\right)+d_{S}\left(k_{u},u\right)\stackrel{\eqref{eq:re}\wedge\eqref{eq:ee}}{\le}O(\log k)\cdot d_{G}(v,u)~.
	\end{equation*}
	We showed a polynomial time algorithm, that given a weighted graph $G=\left(V,E,w\right)$ and a subset $K\subseteq V$ of size $k$, constructs a spanner $H$ with lightness $O(1)$, and such that every pair in $K\times V$ has distortion at most $O(\log k)$.
	Now  \theoremref{thm:reduction_lightness} implies \lemmaref{lem:terminal spanner}.

	\section{A Light Tree with Constant Average Distortion}
	
	Here we prove our main theorem on finding a light spanning tree with constant average distortion. Later on we show a probabilistic embedding into a distribution of light spanning trees with improved bound on higher moments of the distortion.
	\begin{theorem}\label{thm:light scaling tree}
		For any parameter $0<\rho<1$, any graph contains a spanning tree with scaling distortion $\tilde{O}(\sqrt{1/\epsilon})/\rho$ and lightness $1+\rho$.
		\end{theorem}
	
	It follows from \lemmaref{lemma:scaling-q-norm} that the average distortion of the spanning tree obtained is $O(1/\rho)$. Moreover, the
	$\ell_q$-distortion is $O(1/\rho)$ for any fixed $1 \leq q <2$, $\tilde{O}\left(\log^{1.5}n\right)/\rho$ for $q=2$, and $\tilde{O}(n^{1-2/q})/\rho$ for any fixed $2< q < \infty$.

	We will need the following simple lemma, that asserts the scaling distortion of a composition of two maps is essentially the product of the scaling distortions of these maps.\footnote{Note that this is not true for the average distortion -- one may compose two maps with constant average distortion and obtain a map with $\Omega(n)$ average distortion.}

	\begin{lemma}
		\label{lem:Spanners-Composition-Lemma}Let
		$\left(X,d_{X}\right)$, $\left(Y,d_{Y}\right)$ and $\left(Z,d_{Z}\right)$
		be metric spaces. Let $f:X\rightarrow Y$ (respectively, $g:Y\rightarrow Z$) be a non-contractive onto
		embedding with scaling distortion $\alpha$ (resp., $\beta$). Then $g\circ f$
		has scaling
		distortion $\alpha(\epsilon/2)\cdot\beta(\epsilon/2)$.
	\end{lemma}
	\begin{proof}
		Let $n=|X|$. Let $\dist_f(v,u)=\frac{d_{Y}\left(f(v),f(u)\right)}{d_{X}\left(v,u\right)}$ be the distortion of the pair $u,v\in X$ under $f$, and similarly let $\dist_g(v,u)=\frac{d_{Z}\left(g(f(v)),g(f(u))\right)}{d_{Y}\left(f(v),f(u)\right)}$. Fix some $\epsilon\in(0,1)$. We would like to show that at most $\epsilon\cdot{n \choose 2}$ pairs suffer distortion
		greater than $\alpha(\epsilon/2)\cdot\beta(\epsilon/2)$ by $g\circ f$.
		Let $A=\left\{ \left\{ v,u\right\} \in{X \choose 2}:\dist_f(v,u)>\alpha(\epsilon/2)\right\}$
		and
		$B=\left\{ \left\{ v,u\right\} \in{X \choose 2}:\dist_g(v,u)>\beta(\epsilon/2)\right\} $.
		By the bound on the scaling distortions of $f$ and $g$, it holds that $|A\cup B|\le |A|+|B|\le\epsilon\cdot{n \choose 2}$.
		Note that if $\left\{ v,u\right\} \notin A\cup B$ then
		\begin{eqnarray*}
			\frac{d_{Z}\left(g(f(v)),g(f(u))\right)}{d_{X}\left(v,u\right)}&=&\dist_f(v,u)\cdot \dist_g(v,u)\\
			&\le&\alpha(\epsilon/2)\cdot\beta(\epsilon/2)\ ,
		\end{eqnarray*}
		which concludes the proof.
	\end{proof}
	
	We will also need the following result, that was proved in \cite{ABN15}.
	\begin{theorem}[\cite{ABN15}]\label{thm:scaling-spanning-trees}
		Any graph contains a spanning tree with
		scaling distortion $O(\sqrt{1/\epsilon})$.
	\end{theorem}
	
	Now we can prove the main result.
	
	\begin{proof}(of \theoremref{thm:light scaling tree})
		Let $H$ be the spanner given by \theoremref{cor:light spanner scaling distortion}. Let $T$ be a spanning tree of $H$ constructed according to \theoremref{thm:scaling-spanning-trees}. By \lemmaref{lem:Spanners-Composition-Lemma},
		$T$ has scaling distortion  $O(\sqrt{1/\epsilon})\cdot\tilde{O}(\log(1/\epsilon))/\rho=\tilde{O}(\sqrt{1/\epsilon})/\rho$ with respect to the distances in $G$. The lightness follows as $\Psi(T)\le\Psi(H)\le 1+\rho$.
		
	\end{proof}

	\paragraph{Random Tree Embedding.} We also derive a result on probabilistic embedding into light spanning trees with scaling distortion. That is, the embedding construct a distribution over spanning tree so that each tree in the support of the distribution is light. In such probabilistic embeddings \cite{bartal1} into a family ${\cal Y}$, each embedding $f=f_Y:X\to Y$ (for some $(Y,d_Y)\in{\cal Y}$) in the support of the distribution is non-contractive, and the distortion of the pair $u,v\in X$ is defined as $\E_Y\big[\frac{d_Y(f(u),f(v))}{d_X(u,v)}\big]$. The prioritized and scaling distortions are defined accordingly. We make use of the following result from \cite{ABN15}.\footnote{The fact the embedding yields coarse scaling distortion is implicit in their proof.}
	\begin{theorem}(\cite{ABN15})\label{thm:distrib_scaling_span_trees}
		Every weighted graph $G$ embeds into a distribution over spanning trees with coarse scaling distortion $\tilde{O}(\log^2(1/\epsilon))$.
	\end{theorem}
	We note that the distortion bound on the composition of maps in \lemmaref{lem:Spanners-Composition-Lemma} also holds whenever $g$ is a random embedding, and we measure the scaling expected distortion.
	Thus, following the same lines as in the proof of \theoremref{thm:light scaling tree}, (while using \theoremref{thm:distrib_scaling_span_trees} instead of \theoremref{thm:scaling-spanning-trees}), we obtain the following.
	\begin{theorem}\label{thm:light scaling distirubition}
		For any parameter $0<\rho<1$
		and any weighted graph $G$, there is an embedding of $G$  into a distribution over spanning trees with scaling distortion $\tilde{O}(\log^3(1/\epsilon))/\rho$, such that every tree $T$ in the support has lightness $1+\rho$.
	\end{theorem}
	It follows from \lemmaref{lemma:scaling-q-norm} that the $\ell_q$-distortion is $O(1/\rho)$, for every fixed $q\ge 1$.

	\section{Lower Bound on the Trade-off between Lightness and Average Distortion}\label{sec: lower bound}

	In this section, we give an example of a graph for which any spanner with lightness $1+\rho$ has average distortion $\Omega(1/\rho)$ (of course this bound holds for the $\ell_q$-distortion as well). This shows that our results
	are tight \footnote{We also mention that in general the average distortion of a spanner cannot be arbitrarily close to 1, unless the spanner is extremely dense. E.g., when $G$ is a complete graph, any spanner with lightness at most $n/4$ will have average distortion at least $3/2$.}.
	
	\begin{lemma}\label{lem:lwr_bnd}
		For any $n \ge 32$ and $\rho\in [1/n,1/32]$, there is a graph $G$ on $n+1$ vertices such that any spanner $H$ of $G$ with lightness at most $1+\rho$ has average distortion at least $\Omega\left(1/\rho\right)$.
	\end{lemma}
	\begin{proof}

		We define the graph $G=(V,E)$ as follows.
		Denote $V=\left\{ v_{0},v_{1},\dots,v_{n}\right\}$, $E={V \choose 2}$, and the weight function $w$ is defined as follows.
		\[
		w(\left\{ v_{i},v_{j}\right\} )=\begin{cases}
		1 & \mbox{if }\left|i-j\right|=1\\
		2 & \mbox{otherwise }.
		\end{cases}
		\]
		I.e., $G$ is a complete graph of size $n+1$, where the edges $\{v_i,v_{i+1}\}$ have unit weight and induce a  path of length $n$, and all non-path edges have weight $2$. Clearly, the path is the MST of $G$ of weight $n$.
		Let $k=\lceil \rho n \rceil$.
		Let $H$ be some spanner of $G$ with lightness at most $1+\rho \leq \frac{n+k}{n}$, in particular, $w(H)\le n+k$. Clearly $H$ has at least $n$ edges (to be connected). Let $q$ be the number of edges of weight $2$ contained in $H$. Then $w(H)\ge (n-q)\cdot 1 + q\cdot 2=n+q$. Therefore $q \leq k$.
		
		Let $S$ be the set of vertices which are incident on an edge of weight $2$ in $H$. Then $|S|\le 2q \leq 2k$. Let $\delta=\frac{1}{32\rho}$. For any $v\in S$, let $N_v\subseteq V$ be the set of vertices that are connected to $v$ via a path of length at most $\delta$ in $H$, such that this path consists of weight 1 edges only. Necessarily, for any $v\in S$, $|N_v|\le 2\delta+1$. Let $N=\bigcup_{v\in S}N_{v}$, it holds that $|N|\le2k\cdot(2\delta+1) \leq 4\rho n (\frac{1}{16\rho} + 1) \leq \frac{n}{4} + \frac{n}{8} = \frac{3n}{8}$. Let $\bar{N}=V\setminus N$.
		
		Consider $u\in\bar{N}$. By definition of $N$ every weight $2$ edge is further than $\delta$ steps away from $u$ in $H$. It follows that there are at most $2\delta+1$ vertices within distance at most $\delta$ from $u$ (in $H$). Let $F_u=\{v\in V~:~ d_H(u,v)>\delta\}$. It follows that $|F_u|\ge n-2\delta-1$. Note that for any $v\in F_u$, the distortion of the pair $\{u,v\}$ is at least $\frac{\delta}{2}$.  Hence, we obtain that
		\begin{eqnarray*}
			\sum_{\left\{ v,u\right\} \in{V \choose 2}}\frac{d_{H}\left(v,u\right)}{d_{G}\left(v,u\right)}&\ge&\frac{1}{2}\sum_{u\in\bar{N}}\sum_{v\in F_{u}}\frac{d_{H}\left(v,u\right)}{d_{G}\left(v,u\right)}\\
			&\ge&\frac{5n}{16}\cdot\left(n-2\delta-1\right)\cdot\frac{\delta}{2}\\
			&\ge&\frac{5n}{16}\cdot\frac{7n}{8}\cdot\frac{1}{64\rho}\ .
		\end{eqnarray*}
		Finally, the average distortion is bounded as follows.
		\begin{align*}
		\dist_{1}(H) & =\frac{1}{{n+1 \choose 2}}\sum_{\left\{ v,u\right\} \in{V \choose 2}}\frac{d_{H}\left(v,u\right)}{d_{G}\left(v,u\right)}\\
		& \ge\frac{n}{n+1}\cdot\frac{35}{64}\cdot\frac{1}{64\rho}\\
		& \ge\frac{1}{128\rho}~.\qedhere
		\end{align*}
		
	\end{proof}

	\section{Acknowledgments} We are grateful to Michael Elkin and Shiri Chechik for fruitful discussions.

	{\small
		\bibliographystyle{alpha}
		\bibliography{bib}
	}

\end{document}